\documentclass[11pt]{scrartcl} 

\usepackage[utf8]{inputenc}

\usepackage{graphicx}

\usepackage{tabularx}
\usepackage{multirow}

\usepackage{enumerate}

\usepackage{amsmath}
\usepackage{amssymb}
\usepackage{amsfonts}
\usepackage{amsthm}

\usepackage{epstopdf}
\usepackage{float}
\usepackage[caption=false]{subfig}

\usepackage{xcolor}

\newcommand{\calm}{\ensuremath{\mathcal{M}}}
\newcommand{\calf}{\ensuremath{\mathcal{F}}}
\newcommand{\rea}{\ensuremath{\mathbb{R}}}

\DeclareMathOperator{\adj}{adj}
\DeclareMathOperator{\col}{col}

\newtheorem{rem}{Remark}
\newtheorem{prop}{Proposition}

\begin{document}

\title{On-line Identification of Photovoltaic Arrays' Dynamic Model Parameters}

\author{A. Bobtsov$^1$, F. Mancilla-David$^2$, S. Aranovskiy$^3$ and R. Ortega$^4$} 

\date{\footnotesize $^1$Department of Control Systems and Robotics, ITMO University, Kronverkskiy av. 49, Saint Petersburg, 197101, Russia\\ \texttt{bobtsov@mail.ru}\\
	$^2$University of Colorado Denver, Denver, Colorado 80204, USA\\ \texttt{fernando.mancilla-david@ucdenver.edu}\\
	$^3$IETR--CentaleSup\'elec, 35576 Cesson-S\'evign\'e, France\\\texttt{stanislav.aranovskiy@centralesupelec.fr}\\
	$^4$Departamento Acad\'{e}mico de Sistemas Digitales, ITAM, M\'{e}xico\\ \texttt{romeo.ortega@itam.mx}\\ \vspace{0.5cm}
}

\maketitle

\begin{abstract} 
	This paper deals with the problem of on-line identification of the parameters of a realistic dynamical model of a photovoltaic array connected to a power system through a power converter. It has been shown in the literature that, when interacting with switching devices, this model is able to better account for the PV array operation, as compared to the classical five parameter static model of the array. While there are many results of identification of the parameters of the latter model, to the best of our knowledge, no one has provided a solution for the aforementioned more complex dynamic model since it concerns the parameter estimation of a nonlinear, underexcited system with unmeasurable state variables. Achieving such objective is the main contribution of the paper. We propose a new parameterisation  of the dynamic model, which, combined with the powerful identification technique of dynamic regressor extension and mixing, ensures a fast and accurate online estimation of the unknown parameters.  Realistic numerical examples via computer simulations are presented to assess the performance of the proposed approach---even being able to track the parameter variations when the system changes operating point.
\end{abstract}

\textbf{keywords:} {Photovoltaic arrays, on-line parameter estimation, identification of nonlinear systems.}

\section{Introduction}\label{sec1}
Photovoltaic (PV) arrays are quickly becoming an important source of electric power around the world. According to the latest solar industry update reported by the US National Renewable Energy Laboratory \cite{NREL}, 172 GW of PV capacity was added globally in 2021, bringing cumulative capacity to 939 GW. They also report that, last year, 5\% of global electricity generation came from PV. Analysts project continued increases in annual global PV installations for the upcoming years, with estimates suggesting PV could cover a quarter of global electricity needs by mid-century. Models able to properly capture the performance of PV arrays for both planning and operation purposes are thus critical for efficient use of this technology. 

The static behavior of a PV array may be captured via a nonlinear current--voltage ($I$--$V$) characteristic. The single--diode model (SDM) is able to adequately fit the $I$--$V$ {\em static curve}, and is thus widely adopted to represent the performance of a PV array \cite{Masters}. This model makes use of five parameters whose values depend on the solar irradiance ($G$) and the temperature at the PV junction ($T$). Most of the research has focused on developing functional forms seeking to capture $G$ and $T$ dependencies, either considering the physical characteristics of a PV cell \cite{DeSoto2006,boyd2011,tian2012} or using functional approximations from discrete measured points of the $I$--$V$ curve \cite{dobos2019}. A very large literature is availble for the estimation of the parameters of this static curve, see \cite{JOR} for a recent tutorial containing 164 references. Functional forms have been so far {\em unable} to fully capture the behavior of the SDM five parameters over a wide range of operating conditions. To overcome this defficiency we follow \cite {Dynamic2020} and consider a more realistic {\em dynamical model} of the PV array, which is capable of characterizing a PV array under any operating point, removing the need for such functional forms. We take advantage of the fact that in most application PV arrays are interfaced to a power system through a power converter. Because of the switching, a power converter imposes a current (or voltage) with ripple to the PV array. Under these circumstances, the PV array's parasitic capacitance \cite{Krein2013} plays a role on the synthesis of $I$--$V$ characterization, replacing  the static nonlinear curve by a {\em dynamic nonlinear system} with ``orbits" around static operating points---see \cite{Dynamic2016,Dynamic2020} for a more detailed discussion on this matter. In summary, we confront in this paper the task of parameter estimation of a nonlinear, underexcited system with unmeasurable state variables. Providing a solution to this challenging practically relevant problem is the main contribution of the paper.

The remainder of the paper is organized as follows. In Section \ref{sec2} we present the model of the system and formulate the parameter identification problem. The key step in the design is the development of a linear regression equation (LRE) for the system, which is carried out in Section \ref{sec3}. The main result of the paper is given in Section \ref{sec4}. Simulation results of some examples reported in the literature are given in Section \ref{sec5} to illustrate the excellent  {performance} of the proposed estimator.  The paper is wrapped-up with concluding remarks in  Section \ref{sec6}.
%

\section{New Dynamic Model and Formulation of the Parameter Identification Problem} \label{sec2}
In this section we give the new dynamic model used to describe the behavior of the system and formulate the parameter identification problem solved in the paper.

\begin{figure}[ht]
	\centering
	\includegraphics[width=0.8\linewidth]{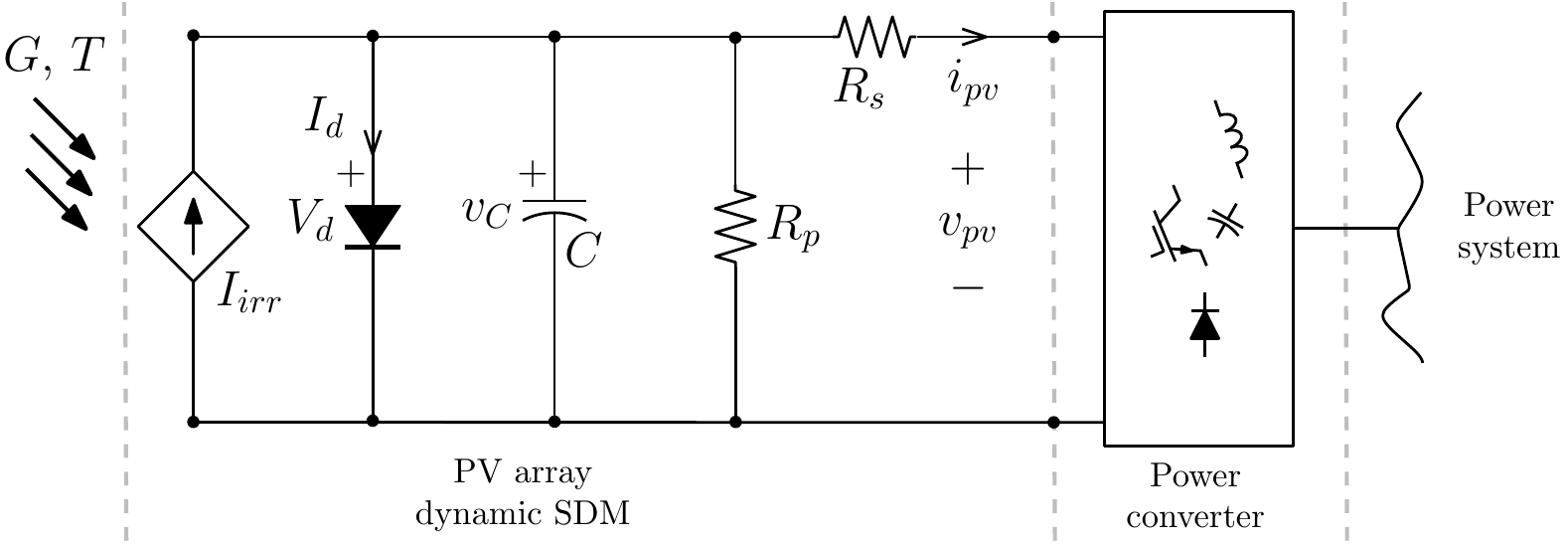}
	\caption{Schematic of PV array connected to a power system through a power converter.}\label{fig:PowerStage}
\end{figure} 

Fig. \ref{fig:PowerStage} shows the circuital representation of the SDM studied in this paper. In the schematic, $v_\mathrm{pv}$ and $i_\mathrm{pv}$ represent, respectively, the voltage and current at the PV array's physical terminals. As suggested in the figure, the SDM includes the photoelectric current or irradiance current generated when the cell is exposed to sunlight, $I_{irr}$, and three parasitic elements, $R_\mathrm{p}$, $R_\mathrm{s}$ and $C$, representing, respectively, the leakage resistance, ohmic losses and the parasitic capacitance, whose dynamics must be added to the model. 

The relation between $I_d$ at $V_d$ at the PV junction is given by the well-known Shockley diode equation,
$$
I_d=I_0\left(e^{b V_d}-1\right),
$$
where $I_0$ is the diode saturation current, and $b$ is the reciprocal of the modified ideality factor, defined as,
$$
b=\frac{q}{nkT},
$$
in terms of the electron's electric charge, $q = 1.602  \times 10^{-19}$ C, the Boltzmann constant, $k = 1.3806503 \times 10^{-23}$ J/K, and the ideality factor of the diode, $n$. The dynamic SDM has thus six characterization parameters, namely $n$, $R_{p},R_s,C,I_{0}$ and $I_{irr}$. However, the nonideality factor $n$ is usually assumed to be constant (independent of $G$ and $T$ variations), and may be obtained from manufacturer's datasheets \cite{tian2012}. Furthermore, sensing $T$ is usually available and changes very slowly compared to electrical dynamics \cite{carrasco2014}, so it can be assumed the parameter $b$ is {\em known and constant}. That leaves five parameters to be identified: $R_{p},R_s,C,I_{0}$ and $I_{irr}$.

It is noted the circuit of Fig. \ref{fig:PowerStage} is able to accommodate a PV array of an arbitrary number of PV cells connected in series/paralell by properly redefining the various characterization parameters \cite{angulo2017}. Fig. \ref{fig:PowerStage} also suggests the PV array is interface to a power system by a generic power converter. In practice, that converter will be either a boost or buck converter topology. In the former case, the PV array would be connected to a series inductor followed by high-frequency switching devices, forcing $i_\mathrm{pv}$ to have a dc value (power component) along with some small ripple (noise component). A buck converter would require a capacitor to be connected in parallel with the PV array, forcing $v_\mathrm{pv}$ to have a dc value along with some small ripple \cite{Mohan2012}. Without loss of generality, we consider a boost--type topological realization.

Noticing from Fig. \ref{fig:PowerStage} that $V_d=v_C$, the state space model for the system considered in the paper is given by the following nonlinear differential-algebraic equations
\begin{subequations}\label{xsys}
	\begin{align}
		\label{x1}
		C\dot v_C & = - {1\over R_\mathrm{p}}v_C - i_\mathrm{pv} + I_{irr} - I_0\left(e^{b v_C}-1\right)\\
		\label{x2}
		v_\mathrm{pv}  & = v_C - R_s i_\mathrm{pv},
	\end{align}
\end{subequations}
with the following considerations:
\begin{enumerate}
	\item[(i)]  $v_\mathrm{pv}(t)\in \rea$ is a positive {\em measurable} signal;
	\item[(ii)] $v_C(t)\in \rea$ is a positive {\em unknown} state variable;
	\item[(iii)] $i_\mathrm{pv}(t)\in \rea$ is a {\em known} signal whose first and second derivative are also {\em known};
	\item[(iv)] $C,R_{p},R_s,I_{0}$ and $I_{irr}$ are positive {\em unknown parameters};
	\item[(v)] $b$ is a {\em known} positive parameter.
\end{enumerate}

\noindent\textbf{Problem formulation.}  Consider the system   \eqref{xsys} verifying the conditions (i)-(v) above.  From the unique measurements of $v_\mathrm{pv}$, $i_\mathrm{pv}$, ${d \over dt}{i}_\mathrm{pv}$ and  ${d^2 \over dt^2} i_\mathrm{pv}$  generate {\em globally convergent}  on-line estimates of the parameters $R_{p},R_s,C,I_{0}$ and $I_{irr}$. \\

All assumptions (i)-(v) are standard and practically reasonable, except perhaps the assumption of knowledge of  ${d \over dt}{i}_\mathrm{pv}$ and  ${d^2 \over dt^2} i_\mathrm{pv}$. However, we follow here the reasoning of \cite[Section 6]{Dynamic2020} regarding some practical considerations pertaining to the shape of $i_\mathrm{pv}$. Namely, that it consists of the sum of a known mean value current plus a ripple of known frequency. Hence, we can assume that $i_\mathrm{pv}(t)=A (1+\sin{\omega t})$ with known $A$ and $\omega$.\\

To simplify the reading we rewrite the system   \eqref{xsys} using control theory notation. Towards this end we define the state, input and output signals via $x:=v_C,u:=i_\mathrm{pv}$ and $y:=v_\mathrm{pv}$ and introduce the constant parameter vector $\eta \in \rea^5$ as
$$
\eta:=\col\left({1 \over R_\mathrm{p}C},{I_{0} \over C},{1 \over C}(I_{irr}+I_0),{1 \over C},R_s\right),
$$
where $\col(\cdot)$ denotes a column vector. With the definitions above the system   \eqref{xsys} may be rewritten as
\begin{subequations} \label{sys}
	\begin{align}
		\label{x}
		\dot x & = - \eta_1 x -\eta_2 e^{b x}+\eta_3-\eta_4 u\\
		\label{y}
		y  & = x - \eta_5 u.
	\end{align}
\end{subequations}

\begin{rem}
\label{rem1}
We make the important observation that the system   \eqref{sys} is {\em nonlinear} with {\em unmeasurable} state. The task of estimating the parameters $\eta$ is, clearly, far from obvious and cannot be solved with any of the existing parameter estimation techniques. Therefore, a radically new technique must be developed to provide the solution to the estimation problem.
\end{rem}

\begin{rem}
\label{rem2}
Notice that it is possible to obtain the physical parameters  $C,R_\mathrm{p},R_s,I_{0}$ and $I_{irr}$ from knowledge of $\eta$. More precisely,
there exists a bijective mapping $\calm:\rea_+^5 \mapsto \rea_+^5$ such that
\begin{equation}
\label{phypar}
K:=\col(C,R_p,R_s,I_{0},I_{irr})=\calm(\eta).
\end{equation}
Clearly, the mapping is defined as
$$
\calm(\eta)=\col\left({1 \over \eta_4},{\eta_4 \over \eta_1},\eta_5,{\eta_2 \over \eta_4},{1 \over \eta_4}(\eta_3 - \eta_2)\right).
$$
\end{rem}

\begin{rem}
\label{rem3}
The interested reader is referred to \cite{BOBetal} where the problem of estimating the parameters of the windmill power coefficient, which has the additional difficulty of being nonlinearly parameterized, is solved. Notice that, if the temperature is {\em not measurable}, that is, if $b$ is unknown, we are confronted with a similar extremely difficult nonlinearly parameterized problem---with the additional constraint of unmeasurable state.
\end{rem}

\section{Model Reparameterization}
\label{sec3}
The key step in the estimator design is to derive a {\em linear regression equation} (LRE) for the parameter $\eta$, that will be used to estimate them. This result is given in the proposition below.

\begin{prop}\label{prop:parametrization}\em
	Consider the system \eqref{xsys}  verifying the conditions (i)-(v).
	\begin{enumerate}
		\item[\textbf{C1}] The system admits the LRE
		\begin{equation}
			\label{lre}
			z(t) =  \Omega^\top(t)\theta+\varepsilon_t,
		\end{equation}
	where $ \Omega(t) \in \rea^8$ and $z(t) \in \rea$ are measurable signals, $\theta \in \rea^8$ is a vector of unknown parameters and $\varepsilon_t$ is an exponentially decaying term stemming from the filters initial conditions.\footnote{Following standard practice, and without loss of generality, this term is neglected in the sequel.}
	\item[\textbf{C2}] There exists a mapping $\calf:\rea_+^4 \mapsto \rea_+^4$ such that
	\[
	\col(\eta_1,\eta_3,\eta_4,\eta_5)=\calf(\theta_1,\theta_2,\theta_3, \theta_4).
	\]
	\item[\textbf{C3}] The parameter $\eta_2$ verifies the relation
	\begin{equation}
	\label{eta2}
	\eta_2=w(\eta_1,\eta_3,\eta_4,\eta_5,u,y)
	\end{equation}
	where $w:\rea^6 \to \rea_+$ is a {\em known} mapping.
	\end{enumerate}
\end{prop}
\begin{proof}
	We will establish the proof considering the representation \eqref{sys} of the system. Differentiating \eqref{y} with respect to time and using \eqref{x} we get
	\[
	\begin{aligned}
		\dot y & =  - \eta_5\dot u+\dot x\\
			& =  - \eta_5\dot u- \eta_1 x -\eta_2 e^{b x}+\eta_3 - \eta_4 u\\
			& =  - \eta_5\dot u- \eta_1 (y +\eta_5 u) -\eta_2 e^{b (y +\eta_5 u)}+\eta_3-\eta_4 u\\
			& =  - \eta_5\dot u-\eta_1 y -\eta_7 u -\eta_2 e^{b y}e^{\eta_6 u}+\eta_3
	\end{aligned}
	\]
	where we defined the positive constants
	$$
	\eta_6:=b \eta_5,\;\eta_7:=\eta_4+\eta_1 \eta_5.
	$$
	Multiplying by $ e^{- b y}$ in both sides of the latter equation we get
	\begin{equation}	\label{eby}
		\begin{aligned}
			e^{- b y}\dot y = &- \eta_5\dot u e^{- b y}-\eta_1 y e^{- b y} -\eta_7 u  e^{- b y} \\ 
				&-\eta_2 e^{\eta_6 u}+\eta_3 e^{- b y}.
		\end{aligned}
	\end{equation}
	Defining now the measurable signal
	\begin{equation} \label{eq:psi}
		\psi:=-{1 \over b} e^{- b y},
	\end{equation}
	we can write \eqref{eby} as
	\begin{equation}
	\label{dotpsi}
	\dot \psi=e^{- b y}\dot y =\phi^\top \mu - \eta_2 e^{\eta_6 u},
	\end{equation}
	where we defined
	\begin{align}
		\mu&:=\col( \eta_5,\eta_1,\eta_7,\eta_3) \nonumber \\
		\phi&:=\col(-\dot u e^{- b y},- y e^{- b y}, -u  e^{- b y},e^{- b y}). \label{eq:phi}
	\end{align}
	Differentiating \eqref{dotpsi} we get
	\[
		\ddot \psi =\dot \phi^\top \mu -\eta_6 \dot u \left(\eta_2 e^{\eta_6 u}\right) =\dot \phi^\top \mu - \eta_6 \dot u\left( \phi^\top \mu-\dot \psi\right).
	\]
	Now, apply to the latter equation an LTI filter of the form ${\lambda^2\over (p+\lambda)^2}$, where $p:={d \over dt}$ and $\lambda>0$ is a designer chosen constant, to get
	\[
		\begin{aligned}
			{\lambda^2 p^2\over (p+\lambda)^2}\psi = \mu^\top {\lambda^2 p\over (p+\lambda)^2} \phi &- \eta_6 \mu^\top {\lambda^2 \over (p+\lambda)^2}( \phi \dot u) \\
				&+ \eta_6 {\lambda^2 \over (p+\lambda)^2} (\dot u \dot \psi).
		\end{aligned}
	\]
	Recalling the definitions of $\eta_6$ and $\mu$, it follows that $\eta_6 = b\mathbf{e}^\top_{1}\mu,$
	where $\mathbf{e}_{1}:= \col(1,0,0,0)$. Then the equation above may be written in the LRE form \eqref{lre} with the definitions
	\begin{equation} \label{eq:parametrisation}
		\begin{gathered}
			z:=\frac{\lambda^2 p^2}{(p+\lambda)^2}\psi,
			\ 
			\theta = \begin{bmatrix}
				\mu \\
				\eta_5\mu
			\end{bmatrix},
			\\
			\Omega:=\begin{bmatrix}
				{\lambda^2 p\over (p+\lambda)^2} \phi  + \frac{\lambda^2}{(p+\lambda)^2} (\dot u \dot \psi)b\mathbf{e}_{1}\vspace{3pt} \\
				-b\frac{\lambda^2}{(p+\lambda)^2}(\phi \dot u)
			\end{bmatrix}.
		\end{gathered}
	\end{equation}
	Since $z$ is clearly measurable, to complete the proof of the claim \textbf{C1} it only remains to prove that the regressor vector $\Omega$ is also measurable. From inspection of $\Omega$ we see that the only conflicting term is ${\lambda^2 \over (p+\lambda)^2} (\dot u \dot \psi)$, which involves the unmeasurable signal $\dot \psi$. To prove that this signal is computable without differentiation we invoke the swapping lemma \cite[Lemma 6.3.5]{SASBODbook} that ensures the following identity
	\[
	\frac{\lambda}{p+\lambda}(\dot u \dot \psi)  = \dot u \frac{\lambda p}{p+\lambda} \psi
	- \frac{1}{p+\lambda}\left(\ddot u \frac{\lambda p}{p + \lambda}\psi \right).
	\]
	Using this identity we obtain
	$$
	{\lambda^2 \over (p+\lambda)^2} (\dot u \dot \psi)={\lambda \over p+\lambda}\left(\dot u \frac{\lambda p}{p+\lambda} \psi
	- \frac{1}{p+\lambda}\left(\ddot u \frac{\lambda p}{p + \lambda}\psi \right)\right),
	$$
	where, given the knowledge of $\ddot u$, the right hand term is computable without differentiation.
	
	We now proceed to prove  the claim \textbf{C3}. For, we apply the filter ${\lambda\over p+\lambda}$ to the state equation \eqref{x} to get
	\[
		\begin{aligned}
			{\lambda\over p+\lambda} \dot x = &-{\lambda\over p+\lambda}[\eta_1(y+\eta_5u)] -\eta_2\left\{{\lambda\over p+\lambda} e^{b(y+\eta_5u)}\right\}
			\\ &+{\lambda\over p+\lambda}(\eta_3-\eta_4 u) = {\lambda p\over p+\lambda}(y+\eta_5u),
		\end{aligned}
	\]
	where we used \eqref{y} to get the second identity. Noticing that the term in brackets is bounded away from zero we can rewrite the equation above as
	\[
		\begin{aligned}
			\eta_2 =&\left({\lambda\over p+\lambda}e^{b(y+\eta_5u)}\right)^{-1}\\
			&\left({\lambda\over p+\lambda}(\eta_3-\eta_4 u - \eta_1(y+\eta_5u))-{\lambda p\over p+\lambda}(y+\eta_5u)\right)\\ 
			& =: w(\eta_1,\eta_3,\eta_4,\eta_5,u,y).
		\end{aligned}
	\]
	
	Finally, the mapping $\calf(\theta_1,\theta_2,\theta_3,\theta_4)$ is easily derived as
	\begin{equation}
	\label{calf}
	\begin{bmatrix} \eta_1 \\ \eta_3 \\ \eta_4 \\ \eta_5 \end{bmatrix}
	= \calf(\theta_1,\theta_2,\theta_3, \theta_4)
	= \begin{bmatrix} \theta_2 \\ \theta_4 \\ \theta_3 - \theta_1\theta_2 \\ \theta_1 \end{bmatrix},
	\end{equation}
	establishing claim \textbf{C2}.
\end{proof}

\begin{rem}
\label{rem4}
The computation of the estimate of $\eta_2$ is done replacing the estimates of $\eta_1,\eta_3,\eta_4$ and $\eta_5$ in the mapping $w$, that is
$$
\hat \eta_2=w(\hat \eta_1,\hat \eta_3,\hat \eta_4,\hat \eta_5,u,y).
$$
For this reason it is important not to ``pull out" the parameter $\eta_1$ from the action of the filter.
\end{rem}
%

\section{Parameter Estimator and Main Result} \label{sec4}

For parameter estimation, we apply the DREM procedure, see \cite{ARAetaltac,ORTetaltac21}. As it is shown in Proposition~\ref{prop:parametrization}, among eight elements of the unknown vector $\theta$, only the first four are required, see \textbf{C2} and \textbf{C3}. The DREM procedure transforms the vector LRE \eqref{lre} into a set of scalar LREs for each of the unknown parameters allowing the estimation of $\theta_{i}$ for $i=1,2,3,4$, only.

For the DREM procedure, we introduce the dynamics extension
\begin{equation} \label{eq:kreff}
	\begin{aligned}
		\dot{\Phi}_f &= -a \Phi_f + a \Omega \Omega^\top, \Phi_f(0) = 0\\
		\dot{Y}_f &= -a Y_f + a \Omega z,  Y_f(0) = 0\\
		\Phi &= c\Phi_f + d\Omega\Omega^\top &&\\
		Y &= cY_f + d\Omega z,&&
	\end{aligned}
\end{equation}
where $z$ and $\Omega$ are defined in \eqref{eq:parametrisation}, $a>0$, $c>0$ and $d>0$ are the tuning coefficients. The dynamic extension \eqref{eq:kreff} is the combination of Kreisselmeiers regressor extension scheme that preserves the excitation of the regressor $\Omega$, see \cite{ARAetaltac22}, and the feedforward term $d\Omega\Omega^\top$ enhancing the excitation, see \cite{ORTetaltac21}.

After the dynamics extension, the LRE
\begin{equation}
\label{kre}
Y(t) = \Phi(t)\theta 
\end{equation}
holds. Following the DREM procedure, we define $\adj\left(\Phi\right)$ to be the adjugate matrix of $\Phi$, and define the signals
\begin{equation}
\label{delcaly}
\Delta := \det\left(\Phi\right),\;\mathcal{Y}:=\adj\left(\Phi\right)Y.
\end{equation} 
Then, upon multiplication of \eqref{kre} by $\adj\left(\Phi\right)$, we obtain the element-wise LREs
\[
\mathcal{Y}_i(t) = \Delta(t)\theta_{i},\;i=1,\ldots,8.
\]
Then, the required unknown parameters are estimated with the classical gradient scheme
\begin{equation} \label{eq:drem:grad}
	\dot{\hat{\theta}}_{i} = -\gamma_i \Delta\left(\mathcal{Y}_i - \Delta\hat\theta_{i}\right),\;i=1,2,3,4
\end{equation}
where $\gamma_i>0$ are tuning coefficients. Note that, thanks to the use of DREM, we can estimate only the required parameters and not the whole vector $\theta$. It is straightforward to show that the estimation error $\tilde{\theta}_i:=\hat{\theta}_i - \theta_i$ obeys
\[
\dot{\tilde{\theta}}_{i} = -\gamma_i \Delta^2\tilde{\theta}_i.
\]
Hence, solving these scalar equations yields
\[
\tilde{\theta}_i(t) = e^{-\gamma_i \int_0^t \Delta^2(s)ds}\tilde{\theta}_i(0).
\]
Then, the convergence follows under the condition that the integral $\int_0^t \Delta^2(s)ds$ tends to infinity that is strictly weaker the the classical condition of persistency of excitation of the regressor $\Omega$, see \cite{ORTetaltac21}.

The calculations above provide the proof of our main result given in the following proposition.

\begin{prop}\label{pro2}\em
	Consider the system \eqref{xsys}  verifying the conditions (i)-(v). Construct the signals $z$ and $\Omega$ as defined in \eqref{eq:parametrisation} with the signals $\psi$ and $\phi$ defined in \eqref{eq:psi} and \eqref{eq:phi}, respectively. Define the DREM parameter estimator \eqref{eq:kreff}, \eqref{delcaly} and \eqref{eq:drem:grad}. Use the estimate $\hat{\theta}_i,\;i=1,\dots,4$, to reconstruct the estimates $\hat \eta$ following \textbf{C2} and \textbf{C3} of Proposition~\ref{prop:parametrization}. Then, compute the estimate of the physical parameters $\hat K:=\col(\hat C,\hat R_p,\hat R_s,\hat I_{0},\hat I_{irr})$ as indicated by \eqref{phypar}. Assume the signal $\Delta$ is not square integrable. Then, 
	$$
	\liminf_{t \to \infty} \big|\hat K(t)-K|=0,
	$$
	for all values of the system and estimator initial conditions. 
\end{prop}

\begin{rem}\label{rem5}
As discussed in  \cite{ARAetaltac,ORTetaltac21} the assumption that $\Delta$ is not square integrable is not a stringent one. As shown in the realistic simulations in the next section it is satisfied in our particular application.
\end{rem}

\section{Illustrative Example} \label{sec5}
As an example, we consider the 85 W Kyocera KC85TS module. We use the functional forms available in \cite{tian2012} to compute the “true” values of characterization parameters at various values of $G$ and $T$. The value of $C$ is generated artificially considering the case studies available in \cite{Dynamic2020,Dynamic2016}. The point of maximum power at standard test conditions (STC), with $G=1000$  W/m$^2$ and $T=25$ $^\mathrm{o}$C, is selected as the main operating point for the identification illustration. Table \ref{tab:params} summarizes the ``true" values of the parameters at STC. 

\begin{table}[tb]
	\caption{``True" PV array parameters for STC and three random operating points.}
	{\newcolumntype{C}{>{}X}
		\setlength{\tabcolsep}{0.3em} 
		\begin{tabularx}{\linewidth}{Cccccc}
			\hline
			\multirow{2}{*}{Description} & \multirow{2}{*}{Symbol} & \multicolumn{4}{c}{Value} \\
			& & STC & Mode  1 & Mode 2 & Mode 3 \\ \hline
			Solar irradiance, W/m$^2$ 			& $G$ 		& 1000 & 748.9 & 740.4 & 715.8 \\
			Junction temperature, K 			& $T$		& 298.15 & 302.15 & 302.40 & 302.71  \\			
			Average input current, A 			& $\bar{i}_\mathrm{pv}$ 	& 4.54 & 3.40 & 3.36 & 3.25 \\	
			Parasitic capacitance, $\mu$F 		& $C$ 		& 0.6 & 0.6 & 0.6 & 0.6\\
			Parallel resistance, $\Omega$ 		& $R_{p}$ 	& 112.55 & 150.28 & 152.02 & 157.23 \\
			Series resistance, $\Omega$ 		& $R_s$ 	& 0.2747 & 0.2747 &0.2747 &0.2747 \\
			Photoelectric current, A 			& $I_{irr}$	& 5.00 	& 3.75 	& 3.70 	& 3.58 \\
			Saturation current, nA 				& $I_{0}$ 	& 10.57 & 17.68 & 18.24 & 18.97 \\	
			Ideality factor, (no units) 		& $n$		& 1.1287 & 1.1287 & 1.1287 & 1.1287 \\
			Exponential coefficient, 1/V 		& $b$ 		& 0.958 & 0.945 & 0.944 & 0.943\\
			\hline
	\end{tabularx}}
	\label{tab:params}
\end{table}

The input signal $u$ in \eqref{sys}, which is $i_\mathrm{pv}$ in \eqref{xsys} is chosen to have an average value corresponding to the point of maximum power, with a 5\%, 20~kHz sinusoidal ripple,
\[
u(t) = \bar{i}_\mathrm{pv}\left(1 + 0.05 \sin( 2\pi\times20\times10^3 t)\right),
\]
where the value $\bar{i}_\mathrm{pv}$ corresponds to the (known) average value of $i_\mathrm{pv}$ as given in Table~\ref{tab:params}.

To implement the estimator we follow Proposition \ref{pro2}. Namely, we construct the signals $z$ and $\Omega$ as defined in \eqref{eq:parametrisation} with the signals $\psi$ and $\phi$ defined in \eqref{eq:psi} and \eqref{eq:phi}, respectively. The filter tuning coefficient is chosen as $\lambda = 6\cdot 10^5$. Then we apply the DREM procedure \eqref{eq:kreff}, \eqref{delcaly} and \eqref{eq:drem:grad}, where the tuning coefficients are selected as $a=10^5$, $c=10^3$, and $d=10^2$. The coefficients $\gamma_i$ are different for different simulation scenarios and are provided below. The estimate $\hat{\theta}$ is further used to construct---applying a certainty equivalence principle---the estimates $\hat \eta$ following \textbf{C2} and \textbf{C3} of Proposition~\ref{prop:parametrization}. Then, estimates of the physical parameters are constructed as indicated in Remark~\ref{rem2} via \eqref{phypar}. As the parameter reconstruction involves algebraic manipulation and division by time-varying signals that can be arbitrary close to zero in transients, the parameter estimates are bounded within the range $[0.1,\ 10]$ times the nominal value. Note that such saturation is used only for illustrative purposes; it is applied for the computation of the physical parameters given in Table~\ref{tab:params} and it is not applied to the estimates $\hat{\theta}$. Thus, the saturation does not affect the convergence analysis of Section \ref{sec4}.

We consider two simulation scenarios. First, the parameter estimator starts with {\em zero initial conditions} and estimates the constant parameters corresponding to the STC operating point. The coefficients $\gamma_i$ are $\gamma_1 = \gamma_2 = 20$ and $\gamma_3 = \gamma_4 = 40$. For this scenario, we present in Fig.~\ref{fig:params} that estimator transients for each of unknown parameters. Whereas the DREM procedure ensures the monotonicity of the transients for $\hat{\theta_i}$, the physical parameters estimates are obtained via algebraic manipulations and are thus prone to chattering. Notably, the estimate $\widehat{I}_0$ has the fastest and chattering-free convergence, the estimate $\widehat{R}_s$ is also chattering-free, whereas the estimate $\widehat{I}_{irr}$ exhibits the strongest chattering. Nevertheless, all transients converge exponentially in $15$~ms to their true values indicated with dotted lines in the figure.
\begin{figure}[tb]
	\centering
	\subfloat[Estimate $\widehat{C}$, $\mu$F]{%
		\includegraphics[width=0.32\linewidth]{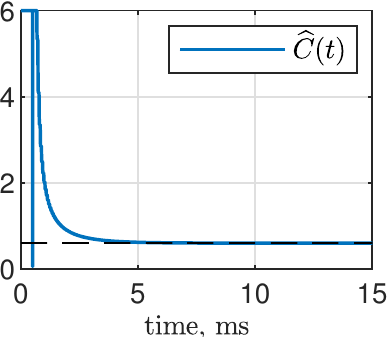}
	}
	~
	\subfloat[Estimate $\widehat{R}_\mathrm{p}$, $\Omega$]{%
		\includegraphics[width=0.32\linewidth]{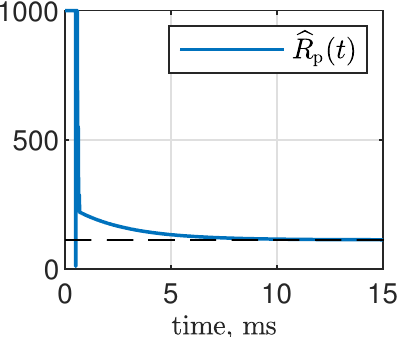}
	}
	~
	\subfloat[Estimate $\widehat{R}_s$, $\Omega$]{%
		\includegraphics[width=0.32\linewidth]{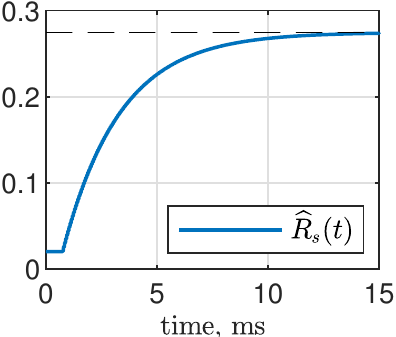}
	}
	\\
	\subfloat[Estimate $\widehat{I}_0$, nA]{%
		\includegraphics[width=0.32\linewidth]{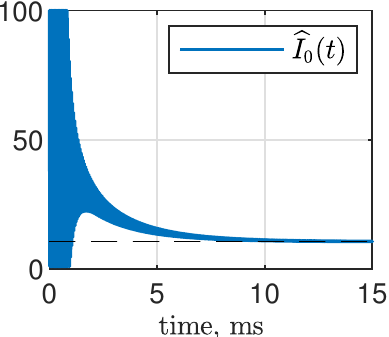}
	}
	~
	\subfloat[Estimate $\widehat{I}_{irr}$, A]{%
		\includegraphics[width=0.32\linewidth]{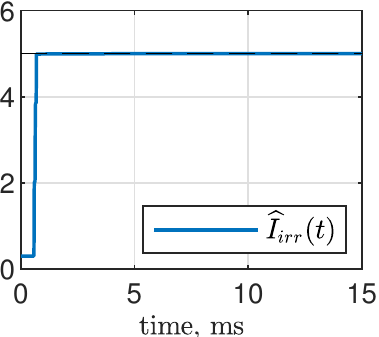}
	}
	\caption{Parameter estimation in the first scenario. The true values are given in dashed lines.}
	\label{fig:params}
\end{figure}

To illustrate the ability of the estimator to {\em track parameter variations}, we also include three additional random operating points on Table~\ref{tab:params} identified as Mode 1, Mode 2 and Mode 3, taken from real measurements of $G$ and $T$. We note the functional forms of \cite{tian2012} assume $R_s$ and $n$ to be constant, independent of $G$ and $T$ variations. Without loss of generality, we assume $C$ to be constant as well. 

So we define a second scenario, where the PV array starts in the Mode~3 with the correct estimates of the parameters. Then, the PV array  switches to Mode~1, then to Mode~2, and then back to Mode~3. The switches occur every~$20$~ms; such a short interval is not practical but is chosen for illustrative purposes. The coefficients $\gamma$ are $\gamma_1 = \gamma_2 = 200$ and $\gamma_3 = \gamma_4 = 400$. For this scenario, we present only the norm of the parameter estimation error $\tilde{\theta}$ to illustrate the overall performance and estimator's capability of tracking time-varying parameters, see Fig.~\ref{fig:tv}. After each switch, the estimation error norm jumps to a value corresponding to the parameter variation between the modes, and then decay exponentially in approximately $15$~ms.

\begin{figure}[tb]
	\centering
	\includegraphics[width=0.49\linewidth]{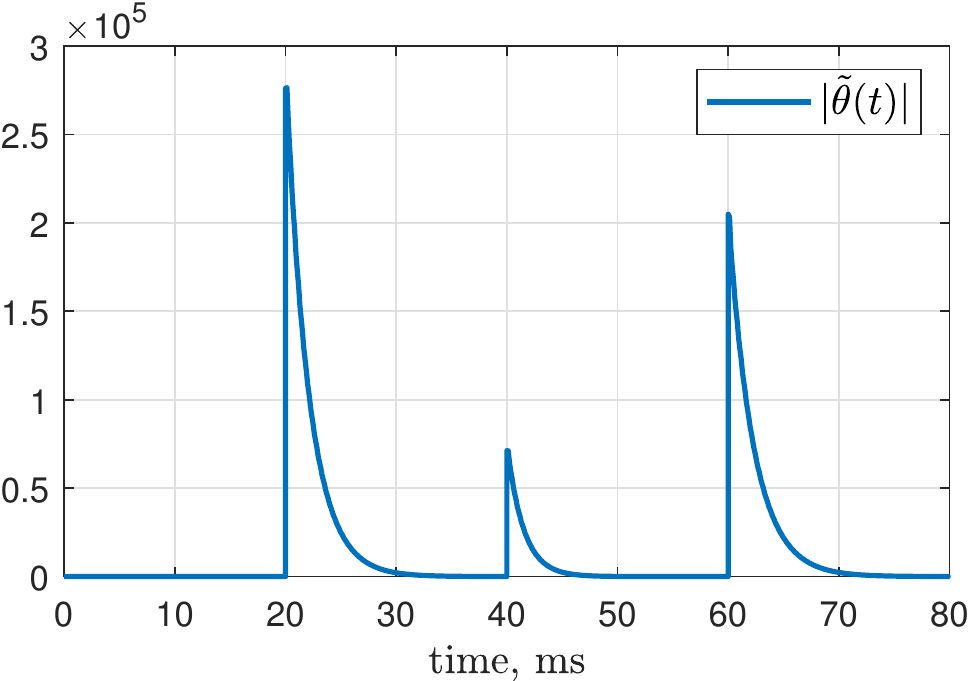}
	\caption{Parameter estimation error norm $|\tilde{\theta}|$ in the second scenario.}
	\label{fig:tv}
\end{figure}

\section{Conclusions} \label{sec6}
In this paper we have provided the first solution to the challenging problem of on-line estimation of the parameters of a new dynamical model describing accurately the behavior of  a PV array connected to a power system through a power converter. This problem concerns the parameter estimation of a nonlinear, underexcited system with unmeasurable state variables. Realistic numerical examples via computer simulations are presented to assess the performance of the proposed approach---even been able to track the parameter variations when the system changes operating point due to its on-line nature. 

We are currently working on the {\em practical} implementation of the proposed identification strategy on a physical PV array and we expect to be able to report our results in the near future.

\bibliographystyle{plain} 
\bibliography{cells_v13_ArXiV}

\begin{thebibliography}{10}

\bibitem{ARAetaltac}
Stanislav Aranovskiy, Alexey Bobtsov, Romeo Ortega, and Anton Pyrkin.
\newblock Performance enhancement of parameter estimators via dynamic regressor
  extension and mixing.
\newblock {\em IEEE Transactions on Automatic Control}, 62(7):3546--3550, jul
  2017.

\bibitem{ARAetaltac22}
Stanislav Aranovskiy, Rosane Ushirobira, Marina Korotina, and Alexey Vedyakov.
\newblock On preserving-excitation properties of kreisselmeiers regressor
  extension scheme.
\newblock {\em IEEE Transactions on Automatic Control}, 2022.

\bibitem{BOBetal}
Alexey Bobtsov, Romeo Ortega, Stanislav Aranovskiy, and Rafael Cisneros.
\newblock On-line estimation of the parameters of the windmill power
  coefficient.
\newblock {\em Systems \& Control Letters}, 164:105242, 2022.

\bibitem{boyd2011}
Matthew~T Boyd, Sanford~A Klein, Douglas~T Reindl, and Brian~P Dougherty.
\newblock Evaluation and validation of equivalent circuit photovoltaic solar
  cell performance models.
\newblock {\em Journal of Solar Energy Engineering}, 133(2), 2011.

\bibitem{angulo2017}
Alejandro~Angulo C{\'a}rdenas, Miguel Carrasco, Fernando Mancilla-David,
  Alexandre Street, and Roberto C{\'a}rdenas.
\newblock Experimental parameter extraction in the single-diode photovoltaic
  model via a reduced-space search.
\newblock {\em IEEE Transactions on Industrial Electronics}, 64(2):1468--1476,
  2016.

\bibitem{carrasco2014}
Miguel Carrasco, Fernando Mancilla-David, and Romeo Ortega.
\newblock An estimator of solar irradiance in photovoltaic arrays with
  guaranteed stability properties.
\newblock {\em IEEE Transactions on Industrial Electronics}, 61(7):3359--3366,
  2013.

\bibitem{DeSoto2006}
Widalys De~Soto, Sanford~A Klein, and William~A Beckman.
\newblock Improvement and validation of a model for photovoltaic array
  performance.
\newblock {\em Solar Energy}, 80(1):78--88, 2006.

\bibitem{dobos2019}
Aron~P Dobos and Janine~M Freeman.
\newblock Significant improvement in pv module performance prediction accuracy
  using a new model based on iec-61853 data.
\newblock Technical report, National Renewable Energy Lab.(NREL), Golden, CO
  (United States), 2019.

\bibitem{NREL}
David Feldman, Krysta Dummit, Jarett Zuboy, Jenny Heeter, Kaifeng Xu, and
  Robert Margolis.
\newblock Spring 2022 solar industry update.
\newblock Technical report, National Renewable Energy Lab.(NREL), Golden, CO
  (United States), 2022.

\bibitem{Dynamic2020}
Yao-Ching Hsieh, Li-Ren Yu, Ting-Chen Chang, Wei-Chen Liu, Tsung-Hsi Wu, and
  Chin-Sien Moo.
\newblock Parameter identification of one-diode dynamic equivalent circuit
  model for photovoltaic panel.
\newblock {\em IEEE Journal of Photovoltaics}, 10(1):219--225, 2019.

\bibitem{JOR}
A~Rezaee Jordehi.
\newblock Parameter estimation of solar photovoltaic (pv) cells: A review.
\newblock {\em Renewable and Sustainable Energy Reviews}, 61:354--371, 2016.

\bibitem{Krein2013}
Katherine~A Kim, Chenyang Xu, Lei Jin, and Philip~T Krein.
\newblock A dynamic photovoltaic model incorporating capacitive and
  reverse-bias characteristics.
\newblock {\em IEEE Journal of Photovoltaics}, 3(4):1334--1341, 2013.

\bibitem{Masters}
Gilbert~M Masters.
\newblock {\em Renewable and efficient electric power systems}.
\newblock John Wiley \& Sons, 2013.

\bibitem{Mohan2012}
Ned Mohan.
\newblock {\em Power Electronics: A First Course}.
\newblock Hoboken, N.J., Wileyn, 2012.

\bibitem{ORTetaltac21}
R.~Ortega, S.~Aranovskiy, A.~A. Pyrkin, A.~Astolfi, and A.~A. Bobtsov.
\newblock New results on parameter estimation via dynamic regressor extension
  and mixing: Continuous and discrete-time cases.
\newblock {\em IEEE Transactions on Automatic Control}, 66(5):2265--2272, 2021.

\bibitem{SASBODbook}
Shankar Sastry and Marc Bodson.
\newblock {\em Adaptive control: Stability, convergence, and robustness}.
\newblock Prentice-Hall, New Jersey, 1989.

\bibitem{tian2012}
Hongmei Tian, Fernando Mancilla-David, Kevin Ellis, Eduard Muljadi, and Peter
  Jenkins.
\newblock A cell-to-module-to-array detailed model for photovoltaic panels.
\newblock {\em Solar Energy}, 86(9):2695--2706, 2012.

\bibitem{Dynamic2016}
Tsung-Hsi Wu, Wei-Chen Liu, Chin-Sien Moo, Hung-Liang Cheng, and Yong-Nong
  Chang.
\newblock An electric circuit model of photovoltaic panel with power electronic
  converter.
\newblock In {\em 2016 IEEE 17th Workshop on control and modeling for power
  electronics (COMPEL)}, pages 1--6. IEEE, 2016.

\end{thebibliography}

\end{document}